\documentclass[10pt]{amsart}
\usepackage[utf8]{inputenc}
\usepackage{amsfonts}
\usepackage{amssymb}
\usepackage{amsmath,amsthm}
\usepackage{amscd}
\usepackage{graphics}
\usepackage{cancel}
\usepackage{pdfsync}
\usepackage[british]{babel} 
\usepackage{mathrsfs}
\usepackage{enumitem}
\usepackage{stmaryrd}
\usepackage{mathrsfs}
\usepackage{wasysym}
\usepackage{latexsym}
\usepackage{mathtools}
\usepackage{afterpage}
\usepackage{threeparttable}
\usepackage{comment}

\usepackage{empheq,etoolbox}
\usepackage{xcolor}
\usepackage[colorlinks,linkcolor=blue,citecolor=blue,urlcolor=blue]{hyperref}

\numberwithin{equation}{section}

\newcommand{\beq}{\begin{equation}}
\newcommand{\eeq}{\end{equation}}
\newcommand{\bea}{\begin{eqnarray}}
\newcommand{\eea}{\end{eqnarray}}
\newcommand{\nn}{\nonumber}
\newcommand\noi{\noindent}
\newcommand{\tbf}{\textbf}

\newcommand{\rd}{\mathrm{d}}

\newcommand{\bk}{\begin{cases}}
\newcommand{\ek}{\end{cases}}
\newcommand{\bepm}{\begin{pmatrix} }
\newcommand{\epm}{\end{pmatrix}}
\newcommand{\bs}{\boldsymbol}

\newcommand{\cH}{\mathcal{H}}

\DeclareMathOperator{\rank}{rank}

\newtheorem{definition}{Definition}
\newtheorem{proposition}{Proposition}
\newtheorem{theorem}{Theorem}

\newtheorem{lemma}{Lemma}
\theoremstyle{definition}

\newtheorem{remark}{\textbf{Remark}}

\newtheorem{example}{\textbf{Example}}

\usepackage{enumerate}
\usepackage{float}
\usepackage{cancel}
\usepackage{hyperref}
\usepackage{mathtools}
\usepackage{afterpage}

\allowdisplaybreaks

\begin{document}

\author{Rafael Azuaje}
\address{Departamento de F\'isica, Universidad Aut\'onoma Metropolitana Unidad Iztapalapa,
San Rafael Atlixco 186, 09340, Ciudad de M\'exico, M\'exico}
\email{razuaje@xanum.uam.mx}
\author{Piergiulio Tempesta}
\address{Departamento de F\'{\i}sica Te\'{o}rica, Facultad de Ciencias F\'{\i}sicas,
  Universidad Complutense de Madrid, 28040 -- Madrid, Spain \\ and Instituto de Ciencias
  Matem\'aticas, C/ Nicol\'as Cabrera, No 13--15, 28049 Madrid, Spain}
\email{p.tempesta@fis.ucm.es, piergiulio.tempesta@icmat.es}

\title[Jacobi-Haantjes manifolds, Integrability and Dissipative Systems]{Jacobi-Haantjes manifolds, Integrability \\ and Dissipative Mechanical Systems}


\date{December 21, 2025}

\begin{abstract}
The notion of Jacobi--Haantjes manifold, consisting of a Jacobi manifold endowed with an algebra of extended Haantjes operator fields,  is proposed as a natural geometric framework which allows us to define the notion of integrability of both conservative and dissipative  Hamiltonian systems, in a unified way. As a reduction, \textit{contact-Haantjes} manifolds are defined.
We prove that the integrability of a contact Hamiltonian system is equivalent to the existence of a suitable Abelian extended Haantjes algebra associated with the system. This result allows us to define a large class of new, completely integrable contact Hamiltonian systems from a given extended Haantjes algebra.

Moreover, we propose a theory of separation of variables for dissipative systems. This result is achieved by lifting a dissipative system into a higher-dimensional manifold, obtained as the symplectization of the Jacobi-Haantjes structure associated with the system. This new manifold naturally acquires the structure of a symplectic-Haantjes manifold. We prove that the Darboux-Haantjes coordinates which separate the Hamilton--Jacobi equation of the higher-dimensional symplectic-Haantjes manifold are in fact separation variables for the Hamilton equations associated with the original dissipative system.
\end{abstract}

\maketitle

\tableofcontents

\section{Introduction}
\label{introduction}

Over the past two decades, there has been renewed interest in the geometry of Nijenhuis and Haantjes tensors and their role in integrable systems. The notion of Nijenhuis torsion was introduced by Nijenhuis \cite{Nij1951} in the context of integrability of eigen-distributions associated with operator fields possessing pointwise distinct eigenvalues. Shortly thereafter, Frölicher and Nijenhuis introduced the graded bracket that now bears their name \cite{FN1956}, which has become a fundamental tool in differential geometry, notably in the theory of almost complex structures as shown by the Newlander–Nirenberg theorem \cite{NN1957,KMS1993}. Building on these ideas, Haantjes introduced the torsion that now carries his name and proved that the vanishing of the Haantjes torsion is a necessary condition for the existence of an integrable frame of generalized eigenvectors \cite{Haa1955}; for pointwise semisimple operators, this condition is also sufficient.
\vspace{2mm}

In recent years, Nijenhuis and Haantjes tensors have found important new applications. They play a central role in the characterization of integrable chains of hydrodynamic-type partial differential equations \cite{FeMa,BogRey}; in the theory of infinite-dimensional integrable systems, and in the study of the WDVV equations and Dubrovin–Frobenius manifolds \cite{MGall13, MGall17}. The notion of a Haantjes algebra, introduced as a free module of Haantjes operators closed under composition \cite{TT2021JGP}, has further enriched this framework. In this context, symplectic--Haantjes manifolds were proposed as an alternative geometric setting for finite-dimensional integrable Hamiltonian systems \cite{TT2022AMPA}, alongside the closely related notion of Poisson--Haantjes manifolds \cite{T2018TMP}. These structures are naturally connected with symplectic–Nijenhuis \cite{MMQUAD1984, FPMPAG2003} and Poisson–Nijenhuis \cite{MMQUAD1984} geometries and are inspired by Magri’s notion of Haantjes manifolds \cite{MGall13}.
\vspace{2mm}

The theory of symplectic-Haantjes manifolds has recently been further developed as a natural framework  for systematically carrying out the separation of variables for the Hamilton--Jacobi equation associated with an integrable conservative system \cite{TT2022AMPA,RTT2022CNS,RTT2024AMPA}.
\vspace{2mm}

This work aims to extend the symplectic-Haantjes geometry in order to provide a unified framework allowing us to discuss both conservative and dissipative systems within a novel, general theoretical setting. Our approach is based on merging Haantjes geometry with the theory of Jacobi manifolds, which generalize several well-known geometric structures, including Poisson, contact, and locally conformal symplectic geometries.
\vspace{2mm}

The results presented in this work are of a twofold nature: on one hand, we introduce new, general geometric structures, which are of intrinsic mathematical interest; on the other hand, we establish theorems relating these structures with Hamiltonian dynamics, with special emphasis on the concepts of complete integrability of dissipative systems.
\vspace{2mm}

More specifically, we define \textit{Jacobi--Haantjes manifolds} by combining Jacobi structures with suitable algebras of extended Haantjes operators. This construction generalizes earlier notions of Jacobi–Nijenhuis manifolds \cite{MMPCRAS1999,PNJGP2003} and extends them to the Haantjes setting. A key ingredient of our approach is the introduction of \textit{extended Haantjes algebras} acting on pairs consisting of vector fields and smooth functions.
\vspace{2mm}

It is well known that transitive Jacobi manifolds reduce either to contact manifolds or locally conformal symplectic (LCS) manifolds, depending on whether their dimension is odd or even \cite{Marle1991}. Hamiltonian mechanics on contact manifolds is particularly suitable for the analysis of dissipative systems \cite{BCTAOP2017,dLS2017}. Contact geometry has found interesting applications (see e.g., \cite{CN1996FP,R2008AOP,G2008,AS}). On the other hand, Hamiltonian mechanics on LCS manifolds studies systems whose dynamics can be locally described by a Hamiltonian system on a symplectic manifold, but it cannot be done globally \cite{esen2021hamilton}.
\vspace{2mm}

A central result of the paper is an extension of the Liouville–Haantjes theorem to the contact setting. We prove that the complete integrability of a contact Hamiltonian system is equivalent to the existence of a suitable Abelian extended Haantjes algebra associated with the system. This result allows us to construct broad new families of completely integrable dissipative Hamiltonian systems in arbitrary dimension.
\vspace{2mm}

Another major contribution of this work is the development of a theory of separation of variables for dissipative systems. By lifting a dissipative Hamiltonian system on a Jacobi-Haantjes manifold to a conservative system on a higher-dimensional symplectic manifold via a Poissonization procedure \cite{TT2016SIGMA,TT2022CMP,RTT2022CNS,RTT2024AMPA}, we obtain a symplectic--Haantjes structure for the lifted dynamics. We show that Darboux–Haantjes coordinates separating the Hamilton–Jacobi equation of the lifted system induce separation variables for the Hamilton's equations of the original dissipative dynamics.
\vspace{2mm}

The framework introduced here provides a unifying perspective bridging conservative and dissipative integrable systems. Several directions for future research naturally arise from this work, including the extension of the present theory to locally conformal symplectic manifolds, the study of cosymplectic and cocontact geometries within the Haantjes framework, and the development of a general theory of reductions for Jacobi--Haantjes manifolds.
\vspace{2mm}

The paper is organized as follows. In Section \ref{sec2}, we review the basic notions of Haantjes geometry and Jacobi structures. In Section \ref{sec3}, we introduce an extension of Haantjes operators suitable for the Jacobi setting. In Section \ref{sec4}, we define the Jacobi--Haantjes manifolds and study their main properties. Section \ref{sec5} focuses on contact–Haantjes manifolds and the integrability of contact Hamiltonian systems. In Section \ref{sec6}, we develop a theory of separation of variables for dissipative systems. Finally, in Section \ref{sec7}, we present new examples of completely integrable dissipative Hamiltonian systems together with their associated Jacobi--Haantjes structures.

\section{Haantjes geometry and Jacobi structures}
\label{sec2}

In this section, we shall revise the basic aspects of the Haantjes geometry.

\subsection{Haantjes operators}
\label{subsec:2.1}
The fundamental notions concerning the geometry of Nijenhuis and Haantjes torsions, which we present here, were originally derived in the papers \cite{Haa1955,Nij1951,FN1956}.
\vspace{2mm}

Let $M$ be a differentiable manifold, $\mathfrak{X}(M)$ the Lie algebra of all smooth vector fields on $M$ and $\boldsymbol{K}:\mathfrak{X}(M)\rightarrow \mathfrak{X}(M)$ be a smooth $(1,1)$ tensor field (namely, an operator field).  From now on the terms ``tensor fields'' and ``operator fields'' will be abbreviated to tensors and operators, respectively. In the following, all tensors are assumed to be smooth. 
\begin{definition} \label{def:N}
The
 \textit{Nijenhuis torsion} of $\boldsymbol{K}$ is the vector-valued $2$-form  defined by
\begin{equation} \label{eq:Ntorsion}
\tau_ {\boldsymbol{K}} (X,Y):=\boldsymbol{K}^2[X,Y] +[\boldsymbol{K}X,\boldsymbol{K}Y]-\boldsymbol{K}\Big([X,\boldsymbol{K}Y]+[\boldsymbol{K}X,Y]\Big),
\end{equation}
where $X,Y \in \mathfrak{X}(M)$ and $[ \ , \ ]$ denotes the Lie bracket of two vector fields.
\end{definition}
\begin{definition} \label{def:H}
 \noi The \textit{Haantjes torsion} of $\boldsymbol{K}$ is the vector-valued $2$-form defined by
\begin{equation} \label{eq:Haan}
\mathcal{H}_{\boldsymbol{K}}(X,Y):=\boldsymbol{K}^2\tau_{\boldsymbol{K}}(X,Y)+\tau_{\boldsymbol{K}}(\boldsymbol{K}X,\boldsymbol{K}Y)-\boldsymbol{K}\Big(\tau_{\boldsymbol{K}}(X,\boldsymbol{K}Y)+\tau_{\boldsymbol{K}}(\boldsymbol{K}X,Y)\Big).
\end{equation}
\end{definition}

\begin{definition}
\label{deHaantjes}
A Haantjes (Nijenhuis)   operator is a (1,1)-tensor  whose  Haantjes (Nijenhuis) torsion identically vanishes.
\end{definition}

A simple, relevant case of Haantjes operator is that of a tensor  $\boldsymbol{K}$  which takes a diagonal form in a local chart $\boldsymbol{x}=(x^1,\ldots,x^n)$:
\begin{equation}
\boldsymbol{K}(\boldsymbol{x})=\sum _{i=1}^n \lambda_{i }(\boldsymbol{x}) \frac{\partial}{\partial x^i}\otimes \rd x^i. \label{eq:Ldiagonal}
 \end{equation}
Here $\lambda_{i }(\boldsymbol{x}):=\lambda^{i}_{i}(\boldsymbol{x})\in C^{\infty}(M)$ are the eigenvalues of $\boldsymbol{K}$ and  $\left(\frac{\partial}{\partial x^1},\ldots, \frac{\partial}{\partial x^n}\right) $ are the fields forming the so-called \textit{natural frame} associated with the local chart $(x^{1},\ldots,x^{n})$. As is well known, the Haantjes torsion of the diagonal operator \eqref{eq:Ldiagonal} vanishes. 
\vspace{2mm}


Haantjes operators exhibit a wealth of algebraic properties, in some cases even richer than those of Nijenhuis operators. Relevant examples of Haantjes operators in classical mechanics and in Riemannian geometry have been proposed, for instance, in \cite{RTT2022CNS},  \cite{RTT2024AMPA}, \cite{TT2022AMPA}, \cite{TT2021JGP}.

\subsection{Haantjes algebras and Haantjes chains}
\label{subsec:2.2}
The crucial concept of Haantjes algebras has been introduced in \cite{TT2021JGP} and further studied and generalized in \cite{RTT2023JNS}. 
\begin{definition}\label{def:HA}
A Haantjes algebra on $M$ is a set $\mathscr{H}$ of Haantjes  operators $\boldsymbol{K}:\mathfrak{X}(M)\to \mathfrak{X}(M)$ that generate:
\begin{itemize}
\item
a free module over the ring of smooth functions on $M$:
\begin{equation}
\label{eq:Hmod}
\mathcal{H}_{\left( f\boldsymbol{K}_{1} + g\boldsymbol{K}_2 \right)}(X,Y)= \mathbf{0}
 \, , \qquad\forall\, X, Y \in \mathfrak{X}(M) \, , \quad \forall\, f,g \in C^\infty(M) \,  , \quad \forall ~\boldsymbol{K}_1,\boldsymbol{K}_2 \in  \mathscr{H} ;
\end{equation}
  \item
a ring  w.r.t. the composition operation:
\begin{equation}
 \label{eq:Hring}
\mathcal{H}_{(\boldsymbol{K}_1 \, \boldsymbol{K}_2)}(X,Y)=\mathbf{0} \, , \qquad
\forall\, \boldsymbol{K}_1,\boldsymbol{K}_2\in  \mathscr{H} , \quad\forall\, X, Y \in \mathfrak{X}(M)\, .
\end{equation}
\end{itemize}
If $\boldsymbol{K}_1\,\boldsymbol{K}_2=\boldsymbol{K}_2\,\boldsymbol{K}_1,$ $\forall~\boldsymbol{K}_1,\boldsymbol{K}_2 \in  \mathscr{H},$
the  algebra $\mathscr{H}$ is said to be an Abelian  Haantjes algebra. 
\end{definition}
In essence, we may regard $\mathscr{H}$ as an associative algebra of Haantjes operators.

\begin{remark}
Abelian Haantjes algebras admit the existence of coordinate systems, called \textit{Haantjes coordinates}, in which all $\boldsymbol{K}\in \mathscr{H}$ can be expressed simultaneously in a block-diagonal form. In particular, if $\mathscr{H}$ consists of semisimple operators, then each $\boldsymbol{K}\in \mathscr{H}$ takes a purely diagonal form in these coordinates \cite{TT2021JGP}. 
\end{remark}

Haantjes chains, in the present formulation, have been introduced in \cite{TT2022AMPA} and will be extensively used in the forthcoming analysis.
\begin{definition} \label{def:7}
 Let $\mathscr{H}$ be a Haantjes algebra on $M$ of rank $m$. We shall say that a smooth function $H$ generates a Haantjes chain $\mathscr{C}$ of closed 1-forms of length $m$ if there exists
a distinguished basis $\{\boldsymbol{K}_1, \ldots, \boldsymbol{K}_m\}$ of $\mathscr{H}$
 such that
\begin{equation} \label{eq:MHchain}
\rd (\boldsymbol{K}^T_i \,\rd H)=\boldsymbol{0} \ , \quad\qquad i=1,\ldots ,m \ ,
\end{equation}
where $\boldsymbol{K}^{T}_{i}: \mathfrak{X}^{*}(M) \to \mathfrak{X}^{*}(M)$ is the transposed operator of $\boldsymbol{K}_{i}$. The (locally) exact 1-forms 
\beq
\rd H_i=\boldsymbol{K}^T_i \,\rd H,
\eeq
which are supposed to be linearly independent, are said to be the elements of the Haantjes chain of length $m$ generated by $H$. The functions $H_i$ are called the potential functions of the chain.
\end{definition}

\subsection{The symplectic-Haantjes manifolds}
\label{subsec:2.3}
If we endow a symplectic manifold with an algebra of Haantjes operators compatible with the symplectic structure, we obtain a symplectic-Haantjes (or $\omega \mathscr{H}$) manifold. Beyond their intrinsic mathematical properties,  these manifolds are significant because they provide a simple yet sufficiently general framework in which the theory of conservative Hamiltonian integrable systems can be rigorously formulated.

\begin{definition}\label{def:oHman}
A symplectic--Haantjes (or $\omega \mathscr{H}$) manifold  of class $m$ is a triple $( M,\omega,\mathscr{H})$ which satisfies the following properties:
\begin{itemize}
\item[(i)]
$(M,\omega)$  is a   symplectic  manifold of dimension $ 2 \, n$;
\item[(ii)]
$\mathscr{H}$ is an Abelian Haantjes algebra of rank $m$;
\item[(iii)]
$(\omega,\mathscr{H})$ are algebraically compatible, that is
\beq
\omega(X,\boldsymbol{K} Y)=\omega(\boldsymbol{K} X,Y) \, , \qquad \forall \boldsymbol{K} \in \mathscr{H} ,
\eeq
or equivalently
\begin{equation}\label{eq:compOmH}
\boldsymbol{\Omega}\, \boldsymbol{K} =\boldsymbol{K}^T\boldsymbol{\Omega} \, ,\qquad \ \forall \boldsymbol{K} \in \mathscr{H} .
\end{equation}
\end{itemize}
\noi
 Hereafter $\boldsymbol{\Omega}:=\omega ^\flat:\mathfrak{X}(M) \rightarrow \mathfrak{X}^{*}(M)$ denotes the  fiber bundles isomorphism defined by
\beq
\omega(X,Y)=\langle \boldsymbol{\Omega} X,Y \rangle \, , \qquad\qquad\forall X, Y \in \mathfrak{X}(M).
\eeq
The map $P:=\boldsymbol{\Omega}^{-1}:\mathfrak{X}^{*}(M) \rightarrow \mathfrak{X}(M)$  is the Poisson bivector induced by  the symplectic two-form $\omega$.
\par
\end{definition}

A crucial result in the theory of $\omega \mathscr{H}$ manifolds is the following characterization of the notion of  integrability in the sense  of Liouville--Arnold.
\begin{theorem}[Liouville--Haantjes \cite{TT2022AMPA}]\label{th:HAL}
Let $M$ be a $2n$-dimensional Abelian $\omega \mathscr{H}$ manifold of class $n$ and $\{H_1,H_2, \ldots, H_n\}$ be smooth potential functions of a Haantjes chain generated by a function $H$.  Then, the foliation generated by these functions is Lagrangian. Consequently,  each Hamiltonian system, with Hamiltonian functions $H$ and $H_\alpha$, $1\le \alpha \le n$, is integrable by quadratures.
Conversely, let us consider  a completely integrable system with $n$ degrees of freedom, defined by a Hamiltonian $H$ and a set of $n$  integrals of motion  $\{H_1,\ldots,H_n \}$ in involution  and independent among each other. Let $\{(J_k,\phi_k)\}$,   $k =1,\ldots, n$, denote a set of action--angle variables, with associated frequencies $\nu_{k}(\boldsymbol{J}):=\frac{\partial{H}}{\partial J_{k}}$. If $H$ is non degenerate, that is
\begin{equation}\label{eq:And}
\det\left(\frac{\partial \nu_k}{\partial J_i} \right)=\det\left(\frac{\partial^2 H}{\partial J_i \partial J_k} \right)\neq 0 \ ,
\end{equation}
then $M$ admits, in any tubular neighborhood of an Arnold  torus, a semisimple $\omega \mathscr{H}$ structure  whose Haantjes algebra is generated by the operators
\beq \label{eq:LAA}
\boldsymbol{K}_\alpha=\sum _{i=1}^n \frac{\nu_i^{(\alpha)}(\boldsymbol{J})}{\nu_i (\boldsymbol{J})}\bigg (\frac{\partial}{\partial J_i}\otimes \rd J_i +\frac{\partial}{\partial \phi_i}\otimes \rd \phi_i \bigg )\quad \alpha=1,\ldots,n
\ ,
 \eeq
 where $\nu_i^{(\alpha)}(\boldsymbol{J})$ are the frequencies of the $(\alpha)-nth$ linear flow.
 \end{theorem}

\subsection{Contact Hamiltonian mechanics}
\label{seccontact}
Hamiltonian systems on contact manifolds naturally describe dissipative systems (for details see \cite{BCTAOP2017,dLS2017}), indeed, the natural phase space for time-independent dissipative Hamiltonian systems is the canonical contact manifold $\mathbb{R}\times T^{*}Q$. Contact systems also have interesting applications in the modern geometric formulation of thermodynamics \cite{BEnt2017}.
For the sake of completeness, we briefly review the formalism of contact Hamiltonian systems.

\begin{definition}
\label{deContact}
A contact manifold is an odd dimensional manifold $M$ equipped with a 1-form $\theta$ such that $\theta\wedge \rd \theta^{n}\neq 0$, where $dim(M)=2n+1$. 
\end{definition}
In the literature, the manifolds in Definition \ref{deContact} are sometimes said to be \textit{co-oriented} contact manifolds \cite{LL2020,G2008,GrGr2022}.



\vspace{2mm}
Let $(M,\theta)$ be a contact manifold. For each $f\in C^{\infty}(M)$, there is a vector field $X_{f}\in \mathfrak{X}(M)$, called the \textit{Hamiltonian vector field} for $f$, which is the only one vector field on $M$ such that
\begin{equation}
X_{f}\lrcorner \theta = -f \hspace{1cm}\rm{and}\hspace{1cm} X_{f}\lrcorner d\theta =df-(Rf)\theta,
\end{equation}
where $R$ is the Reeb vector field which obeys 
\begin{equation}
R\lrcorner \theta =1 \hspace{1cm}\rm{and}\hspace{1cm} R\lrcorner d\theta =0.
\end{equation}
Each contact structure on a manifold defines a Lie bracket on the space of smooth functions on such a manifold. Indeed, the bracket defined for any pair of functions $f,g\in C^{\infty}(M)$ by
\begin{equation}
\lbrace f,g\rbrace=X_{g}f+fRg,
\end{equation}
is a Lie bracket on $C^{\infty}(M)$. In addition this bracket fulfills the weak Leibniz rule, namely
$supp(\lbrace f,g\rbrace)\subseteq supp(f)\cap supp(g)$. A Lie bracket satisfying the weak Leibniz rule is called a \textit{Jacobi bracket}. Manifolds endowed with a Jacobi bracket are called Jacobi manifolds; thus, every contact manifold is a Jacobi manifold. Formally we have the following
\begin{definition}
A Jacobi manifold is a triple $(M,\Lambda,E)$ where $M$ is a smooth manifold, $\Lambda$ is a skew-symmetric contravariant smooth 2-tensor field on $M$ and $E$ is a smooth vector field on $M$ such that
\begin{equation}
[\Lambda,\Lambda]_{SN}=2E\wedge \Lambda \qquad\text{and}\qquad [\Lambda,E]_{SN}=0.
\end{equation}
Here $[\cdot,\cdot]$ denotes the Schouten-Nijenhuis bracket on $M$. The pair $(\Lambda,E)$ is said to be a Jacobi structure on $M$. 
\begin{remark}
If $E=0$, then $[\Lambda,\Lambda]=0$ and a Jacobi manifold $(M,\Lambda,E)$ reduces to a Poisson manifold $(M,P)$ with the identification $P=\Lambda$. 
\end{remark}
\end{definition}
Each Jacobi structure $(\Lambda,E)$ induces a Jacobi bracket on $C^{\infty}(M)$ defined  for all $f,g\in C^{\infty}(M)$ by
\begin{equation} \label{eq:2.16}
\lbrace f,g\rbrace=\Lambda(\rd f,\rd g)+fEg-gEf.
\end{equation}
Conversely, any Jacobi bracket arises from a Jacobi structure \cite{ILM1997,Marle1991}. 
\vspace{2mm}

Given a Jacobi manifold $(M,\Lambda,E)$, for each $f\in C^{\infty}(M)$ there is one and only one smooth vector field $X_{f}$ on $M$ such that
\begin{equation}
\lbrace g,f\rbrace=X_{f}g+gEf \, \forall g\in C^{\infty}(M).
\end{equation}
$X_{f}$ is called the Hamiltonian vector field for $f$. In terms of the induced map $\Lambda^{\sharp}:\Omega^{1}(M)\longrightarrow \mathfrak{X}(M)$ given by $\beta(\Lambda^{\sharp}(\alpha))=\Lambda(\beta,\alpha)$, we have 
\begin{equation}
X_{f}=\Lambda^{\sharp}(df)-fE.
\end{equation}

The fundamental notions of Jacobi and Poisson manifolds were introduced by Lichnerowicz in \cite{Lich77,Lich78}. The theory of Jacobi structures has been further developed in \cite{Marle1991,ILM1997,LLMP1999,Esen2022,
Vais2002}.
\vspace{2mm}

The Jacobi structure $(\Lambda,E)$ on $M$ defined by the contact form $\theta$ is
\begin{equation}
\Lambda(\alpha,\beta)=\rd \theta(\sharp(\alpha),\sharp(\beta))\quad\text{and}\quad E=\sharp(\theta),
\end{equation}
where $\sharp$ is the inverse of the map $\flat:\mathfrak{X}(M)\longrightarrow\Omega^{1}(M)$ defined by 
\begin{equation}
\flat(X)=\iota_{X} \rd \theta +(\iota_{X} \theta) \theta.
\end{equation}
Observe that $E=R$.
\vspace{2mm}

Each contact manifold admits sets of Darboux coordinates. In fact, around any point $x\in M$ there exist local coordinates $(q^{1},\ldots,q^{n},p_{1},\ldots,p_{n},z)$, called canonical coordinates, such that
\begin{equation}
\theta=dz-\sum_{i=1}^{n}p_{i}dq^{i}.
\end{equation}
Thus, $\Lambda$ and $E$ (or $R$) can be expressed in local form as
\begin{equation}
\Lambda=\sum_{i=1}^{n}\left(\frac{\partial}{\partial q^{i}}+p_{i}\frac{\partial}{\partial z}\right)\wedge \frac{\partial}{\partial p_{i}}\quad\text{and}\quad E=\frac{\partial}{\partial z}.
\end{equation}
The Hamiltonian vector field $X_{f}$ for $f$ has the local form
\begin{equation}
X_{f}=\sum_{i=1}^{n}\left(\frac{\partial f}{\partial p_{i}}\frac{\partial}{\partial q^{i}}-\left( \frac{\partial f}{\partial q^{i}}+p_{i}\frac{\partial f}{\partial z}\right) \frac{\partial}{\partial p_{i}}\right)+\left(\sum_{i=1}^{n}p_{i}\frac{\partial f}{\partial p_{i}}-f\right)\frac{\partial}{\partial z},
\end{equation}
and the local expression for the Jacobi bracket of two functions $f,g\in C^{\infty}(M)$ is
\begin{equation}
\lbrace f,g\rbrace=\sum_{i=1}^{n}\left(\frac{\partial f}{\partial q^{i}}\frac{\partial g}{\partial p_{i}}-\frac{\partial f}{\partial p_{i}}\frac{\partial g}{\partial q^{i}}\right)+\frac{\partial f}{\partial z}\left(\sum_{i=1}^{n}p_{i}\frac{\partial g}{\partial p_{i}}-g\right)-\frac{\partial g}{\partial z}\left(\sum_{i=1}^{n}p_{i}\frac{\partial f}{\partial p_{i}}-f\right).
\end{equation}

Given a function $H\in C^{\infty}(M)$, $X_{H}$ defines Hamiltonian dynamics on $M$. In fact, the equations of motion in canonical coordinates for the dynamics defined by $X_{H}$ are the so-called \textit{contact Hamilton's equations} \cite{LL2019,dLS2017,BCTAOP2017}, namely
\begin{equation}
\dot{q^{i}} =\frac{\partial H}{\partial p_{i}}, \hspace{1cm}
\dot{p_{i}} =-\left( \frac{\partial H}{\partial q^{i}}+p_{i}\frac{\partial H}{\partial z}\right) , \hspace{1cm}
\dot{z}=p_{i}\frac{\partial H}{\partial p_{i}}-H.
\end{equation}
In this case, we will say that $(M,\theta,H)$ is a contact Hamiltonian system with Hamiltonian function $H$. The evolution of a function $f\in C^{\infty}(M)$ (an observable) along the trajectories of the system reads
\begin{equation}
\dot{f}=L_{X_{H}}f=X_{H}f=\lbrace f,H\rbrace-fRH.
\end{equation}
A function $f$ is a constant of motion if $L_{X_{H}}f=0$ (equivalently, $\lbrace f,H\rbrace-fRH=0$). The Hamiltonian function $H$ is not a constant of motion, indeed,
\begin{equation}
\dot{H}=-HRH.
\end{equation}
It is said that $H$ is a \textit{dissipated quantity} \cite{BG2023,GGMRR2020,LL2020}.

\section{An extension of the theory of Haantjes operators}
\label{sec3}

The notion of \textit{Jacobi-Nijenhuis manifolds} was introduced in \cite{MMPCRAS1999}. Another, stricter definition was proposed in \cite{PNJGP2003}. The main motivation for the latter construction is that it generalizes in a natural way the notion of Poisson-Nijenhuis manifolds introduced by Magri and Morosi in \cite{MMQUAD1984} in the study of Hamiltonian integrable systems. In both constructions \cite{MMPCRAS1999} and \cite{PNJGP2003}, the notion of Nijenhuis torsion has to be extended conveniently.
The aim of this section is to extend these ideas to the more general setting of Haantjes geometry.

\vspace{2mm}
Let $(M, \Lambda, E)$ be a Jacobi manifold. Following \cite{PNJGP2003}, we consider the quadruple $(\bs{K}, Y, \gamma, k)$, where $\bs{K}$ is a (1,1)-tensor field, $Y\in \mathfrak{X}(M)$, $\gamma\in \Omega^1(M)$ and $k\in C^{\infty}(M)$. Out of these objects we construct a new linear map. Precisely, let $\bs{\mathcal{K}}: \mathfrak{X}(M)\times C^{\infty}(M, \mathbb{R})\to \mathfrak{X}(M)\times C^{\infty}(M, \mathbb{R})$ be the $C^{\infty}(M, \mathbb{R})$-linear map defined by
\beq \label{eq:3.1}
\bs{\mathcal{K}}(X,f):= (\bs{K}X+fY, \gamma(X) +kf)
\eeq
for all pairs $(X,f)\in \mathfrak{X}(M) \times C^{\infty}(M, \mathbb{R})$.  Note that the space $\mathfrak{X}(M)\times C^{\infty}(M, \mathbb{R})$ converts into a real Lie algebra once endowed with the bracket
$[,]: \mathfrak{X}(M)\times C^{\infty}(M, \mathbb{R}) \times \mathfrak{X}(M)\times C^{\infty}(M, \mathbb{R})$ defined by
\beq
[(X,f),(Z,h)]= \big([X,Z], \rd h(X) - \rd f(Z)\big),
\eeq
for all $(X,f), (Z,h)\in \mathfrak{X}(M)\times C^{\infty}(M, \mathbb{R})$.
Let us also introduce the $C^{\infty}(M, \mathbb{R})$ map $(\alpha, f): \mathfrak{X}(M) \times C^{\infty}(M, \mathbb{R})\to C^{\infty}(M)$ defined by
\beq
(\alpha, f)(X,h):= \alpha(X)+ fh.
\eeq
The  transposed map
$\bs{\mathcal{K}}^{T}: \Omega^{1}(M) \times C^{\infty}(M, \mathbb{R})\to  \Omega^{1}(M) \times C^{\infty}(M, \mathbb{R})$ is defined by
\beq
(\bs{\mathcal{K}}^{T}(\alpha, f))(X,h):= (\alpha, f) \bs{\mathcal{K}}(X,h).
\eeq
Explicitly, we have
\beq
(\bs{\mathcal{K}}^{T}(\alpha, f))(X,h)=(\alpha, f)(\bs{\mathcal{K}} X+ hY, \gamma(X)+k h)=\alpha(\bs{K} X+hY)+(\gamma(X)+kh) f.
\eeq
\begin{definition}\cite{MMPCRAS1999}
The Nijenhuis torsion $\tau(\bs{\mathcal{K}})$ of $\bs{\mathcal{K}}$ is the $C^{\infty}(M, \mathbb{R})$-bilinear map given by
\bea
\nn \tau(\bs{\mathcal{K}}) ((X,f),(Z,h))&:=&[\bs{\mathcal{K}}(X,f), \bs{\mathcal{K}}(Z,h)]- \bs{\mathcal{K}} [\bs{\mathcal{K}}(X,f), (Z,h)]- \bs{\mathcal{K}} [(X,f), \bs{\mathcal{K}}(Z,h)]\\ \nn &+& \bs{\mathcal{K}}^2 [(X,f), (Z,h)]\\
\eea
for all $(X,f), (Z,h)\in \mathfrak{X}(M)\times C^{\infty}(M, \mathbb{R})$.
\end{definition}
By analogy, we propose the new
\begin{definition}
Let $\bs{\mathcal{K}}: \mathfrak{X}(M)\times C^{\infty}(M, \mathbb{R})\to \mathfrak{X}(M)\times C^{\infty}(M, \mathbb{R})$ be the linear map \eqref{eq:3.1}. The Haantjes torsion of $\bs{\mathcal{K}}$ is the $C^{\infty}(M, \mathbb{R})$-bilinear map given by
\bea
\nn \cH(\bs{\mathcal{K}}) ((X,f),(Z,h))&:=&\bs{\mathcal{K}}^{2}\tau(\bs{\mathcal{K}})((X,f),(Z,h))+\tau(\bs{\mathcal{K}})(\bs{\mathcal{K}}(X,f),\bs{\mathcal{K}}(Z,h))\\
&-& \bs{\mathcal{K}} \big(\tau(\bs{\mathcal{K}})((X,f),\bs{\mathcal{K}}(Z,h))+\tau(\bs{\mathcal{K}})(\bs{\mathcal{K}}(X,f),(Z,h))\big)
\eea
for all $(X,f),(Z,h)\in \mathfrak{X}(M)\times C^{\infty}(M, \mathbb{R})$.
\end{definition}
\begin{definition}
We shall say that $\bs{\mathcal{K}}$ is a Haantjes (Nijenhuis) operator on $\mathfrak{X}(M)\times C^{\infty}(M, \mathbb{R})$ if its Haantjes (Nijenhuis) torsion vanishes.
\end{definition}

The following result provides a geometric characterization of the Haantjes operators. 
\begin{proposition} \label{prop:1}
\label{prop1}
$\bs{\mathcal{K}}=(\bs{K}, Y, \gamma, k)$ is a Haantjes operator on $\mathfrak{X}(M)\times C^{\infty}(M, \mathbb{R})$ if and only if the following relations hold:
\begin{equation}
\label{eqHaantjesex}
\left\lbrace\begin{array}{c}
\mathcal{H}_{\bs{K}}=Y\otimes\gamma \\
\mathcal{L}_{\bs{K}}^{\tau_{\bs{K}}}\gamma=kd\gamma\\
\mathcal{L}_{Y}^{\tau_{\bs{K}}}\bs{K}=-Y\otimes dk\\
\bs{K}^{T}dk=\mathcal{L}_{Y}^{\tau_{\bs{K}}}\gamma+kdk,
\end{array}\right.
\end{equation}
where we have used the notation
\begin{equation}
(\mathcal{L}_{\bs{K}}^{\tau_{\bs{K}}}\gamma)(X,Z)=(\bs{K}X)(\gamma(Z))-(\bs{K}Z)(\gamma(X))-\gamma(\tau_{\bs{K}}(\bs{K}X,Z))-\gamma(\tau_{\bs{K}}(X,\bs{K}Z))+\gamma(\bs{K}(\tau_{\bs{K}}([X,Z]))),
\end{equation}
\begin{equation}
(\mathcal{L}_{Y}^{\tau_{\bs{K}}}\bs{K})(X)=\bs{K}(\tau_{\bs{K}}(X,Y))-\tau_{\bs{K}}(\bs{K}X,Y)
\end{equation}
and
\begin{equation}
(\mathcal{L}_{Y}^{\tau_{\bs{K}}}\gamma)(X)=Y(\gamma(X))-\gamma(\tau_{\bs{K}}(Y,X)).
\end{equation}
Besides, system (\ref{eqHaantjesex}) is equivalent to the requirement that the operator 
\begin{equation}
\label{Ktilde}
\tilde{\bs{K}}=\bs{K}+Y\otimes dt+\frac{\partial}{\partial \tau}\otimes\gamma+k\frac{\partial}{\partial \tau}\otimes d\tau
\end{equation}
on $M\times \mathbb{R}$ is a Haantjes operator (in the sense of Definition \ref{deHaantjes}), where $\tau$ is the global coordinate on $\mathbb{R}$. Thus, $\bs{\mathcal{K}}=(\bs{K}, Y, \gamma, k)$ is a Haantjes operator on $\mathfrak{X}(M)\times C^{\infty}(M, \mathbb{R})$ if and only if $\tilde{\bs{K}}$ is a Haantjes operator on $M\times\mathbb{R}$.
\end{proposition}
The proof of Proposition \ref{prop:1}, \textit{mutatis mutandis}, closely follows the argument presented in \cite{nunes1999some,PNJGP2003} for the Nijenhuis case, and will therefore not be detailed here. In essence, it consists of replacing the Lie bracket $[\cdot, \cdot]$ of vector fields by the Nijenhuis torsion $\tau_{\bs{K}}$ of $\bs{K}$, thereby extending the construction to encompass the Haantjes torsion. It is worth remarking that the Nijenhuis torsion $\tau_{\bs{K}}$ of $\bs{K}$ is not a Lie bracket, as it does not satisfy the Jacobi identity. Nevertheless, it is a $C^{\infty}(M)$-bilinear map on the field of vector fields, which is sufficient to carry out the necessary computations leading to the stated result.
\vspace{2mm}

It is natural to propose an extension of the notion of Haantjes algebras given in Definition \ref{def:HA} to the case of Haantjes operators on  $\mathfrak{X}(M)\times C^{\infty}(M, \mathbb{R})$.
\begin{definition}
An extended Haantjes algebra on $M$ is a set $\mathscr{H}$ of Haantjes operators on $\mathfrak{X}(M)\times C^{\infty}(M, \mathbb{R})$ such that $\forall \, \bs{\mathcal{K}}_1, \bs{\mathcal{K}}_2\in \mathscr{H}$:
\begin{equation}
\label{eq:Hmodextend}
\cH\left( \lambda_1\bs{\mathcal{K}}_{1} +\lambda_2\bs{\mathcal{K}}_2 \right)((X,f),(Z,h))= \mathbf{0}
 \, , \qquad\forall\, X, Z \in \mathfrak{X}(M) \, , \quad \forall\, \lambda_1, \lambda_2, f,h  \in C^\infty(M);
\end{equation}
\begin{equation}
 \label{eq:Hringextend}
\cH(\bs{\mathcal{K}}_1 \, \bs{\mathcal{K}}_2)((X,f),(Z,h)=\mathbf{0} \, ,  \quad\forall\, X, Z \in \mathfrak{X}(M)\, , \quad \forall\, f,h  \in C^\infty(M).
\end{equation}
\end{definition}

\section{The theory of Jacobi-Haantjes manifolds}
\label{sec4}

In this section we introduce the central notion of our construction. To this aim, let $(M,\Lambda,E)$ be a Jacobi manifold, we denote by $(\Lambda, E)^{\sharp}: T^{*}M\times \mathbb{R}\to TM\times \mathbb{R}$ the vector bundle map defined by
\beq
(\Lambda, E)^{\sharp}(\alpha, f):= (\Lambda^{\sharp}(\alpha)-fE, \alpha(E)).
\eeq
where $\Lambda^{\sharp}:\Omega^{1}(M)\longrightarrow\mathfrak{X}(M)$ is the induced map given by  $\langle \beta, \Lambda^{\sharp}(\alpha)\rangle= \Lambda(\beta,\alpha)$.

\subsection{Jacobi-Haantjes structures and Haantjes chains}
\begin{definition}
\label{def:JHM}
A Jacobi-Haantjes (JH) manifold of class $m$ is a quadruple $(M, \Lambda, E, \mathscr{H})$, where $\mathscr{H}$ is an Abelian extended Haantjes algebra on $M$ of rank $m$ such that
\beq
\label{eqcompatibility}
\bs{\mathcal{K}} \circ (\Lambda, E)^{\sharp}= (\Lambda, E)^{\sharp}\circ \bs{\mathcal{K}}^{T}, \qquad \forall \bs{\mathcal{K}} \in \mathscr{H}.
\eeq
\end{definition}
\begin{remark}
The notion of Jacobi-Haantjes structures, as defined above, can be relaxed. Rather than requiring that the Haantjes torsion $\mathcal{H}(\bs{\mathcal{K}})$ identically vanishes on $M$, one can require that it vanishes on the image of $(\Lambda, E)^{\sharp}$. This broader notion, in the Haantjes scenario, parallels that of a Jacobi-Nijenhuis manifold, as proposed by Marrero et al. in \cite{MMPCRAS1999}. Our formulation of Jacobi-Haantjes manifolds in Definition \ref{def:JHM} is more similar to the Jacobi-Nijenhuis structures proposed by Petalidou et al. in \cite{PNJGP2003}, which we generalize directly.
\end{remark}

Given the extended operator $\bs{\mathcal{K}}=(\bs{K}, Y, \gamma, k)$, the compatibility condition \eqref{eqcompatibility} is equivalent to the following system of equations
\begin{equation}
\label{syscompatibility}
\left\lbrace \begin{array}{c}
\bs{K}\circ \Lambda^{\sharp}+E\otimes Y=\Lambda^{\sharp}\circ \bs{K}^{T}-Y\otimes E,\\ 
\Lambda^{\sharp}(\gamma)=-\bs{K}E-kE,\\ 
\gamma(E)=0.
\end{array} \right. 
\end{equation}
Equivalently, the relations
\begin{equation}
\label{systemcc}
\left\lbrace \begin{array}{c}
\Lambda(\bs{K}_{i}^{T}\alpha,\beta)+\alpha(Y_{i})\beta(E)=\Lambda(\alpha,\bs{K}_{i}^{T}\beta)-\beta(Y_{i})\alpha(E), \\ 
\Lambda(\gamma_{i},\alpha)=(\bs{K}_{i}^{T}\alpha)(E)-k_{i}\alpha(E), \\ 
\gamma_{i}(E)=0,
\end{array} \right.
\end{equation}
hold for all $\alpha,\beta\in\Omega^{1}(M)$.
\vspace{2mm}

By analogy with Definition \ref{def:7}, we extend the notion of Haantjes chain to the case of an extended Haantjes algebra.
\begin{definition}
Let $\mathscr{H}$ be an extended Haantjes algebra on $M$ of rank $m$. We shall say that a function $H\in C^{\infty}(M)$ generates an extended Haantjes chain $\mathscr{C}$ on $M$ of length $m$, if  there exists a distinguished basis $\{\bs{\mathcal{K}}_1,\ldots, \bs{\mathcal{K}}_m\}$ of $\mathscr{H}$ and a set of functions $H_{1},\ldots,H_{m}\in C^{\infty}(M)$ such that
 \beq \label{eqHCportuguesa}
 (\rd H_{i}, H_{i})= \bs{\mathcal{K}}_i^{T}(\rd H,H), \qquad i=1,\ldots, m.
 \eeq
The functions $H_i$ are said to be the potential functions of the chain.
\end{definition}

One of the main results of the theory of Jacobi-Haantjes manifolds is the following
\begin{theorem}
\label{themain}
Let $(M,\Lambda,E,\mathscr{H})$ be a Jacobi-Haantjes manifold of class $m$. If $H\in C^{\infty}(M)$ generates an extended Haantjes chain $\mathscr{C}$ on $M$ of length $m$,
then the potential functions $H_{1},\ldots,H_{m}$ of the chain are independent functions pairwise in involution such that $\lbrace H_{i},H\rbrace=0$. 
\end{theorem}

\begin{proof}
Let $H_{1},\ldots,H_{m}$ be the potential functions of the chain $\mathscr{C}$ and $\lbrace \bs{\mathcal{K}}_{1},\ldots,\bs{\mathcal{K}}_{m}\rbrace$ be a distinguished basis of $\mathscr{H}$ such that equation (\ref{eqHCportuguesa}) is satisfied. 
As before, $\bs{\mathcal{K}}_{i}=(\bs{K}_{i},Y_{i},\gamma_{i},k_{i})$, where $\bs{K}_{i}$ is a (1,1)-tensor field on $M$, $Y_{i}\in \mathfrak{X}(M)$, $\gamma_{i}\in \Omega^{1}(M)$ and $k_{i}\in C^{\infty}(M)$. Thus, for all $(X,f)\in\mathfrak{X}(M)\times C^{\infty}(M)$ we obtain
\begin{equation}
\begin{split}
\left( \bs{\mathcal{K}}_{i}^{T}(\rd H,H)\right)(X,f)&=(\rd H,H)\left( \bs{\mathcal{K}}_{i}(X,f)\right)\\
&=(\rd H,H)(\bs{K}_{i}X+fY_{i},\gamma_{i}X+k_{i}f)\\
&=\rd H(\bs{K}_{i}X+fY_{i})+(\gamma_{i}X+k_{i}f)H\\
&=(\bs{K}_{i}^{T}\rd H+H\gamma_{i})X+(Y_{i}H+Hk_{i})f.
\end{split}
\end{equation}
Besides,
\begin{equation}
(\rd H_{i},H_{i})(X,f)=\rd H_{i}(X)+H_{i}f,
\end{equation}
therefore
\begin{equation}
\label{dHi,Hi}
\rd H_{i}=\bs{K}_{i}^{T}\rd H+H\gamma_{i}\quad\textit{and}\quad H_{i}=Y_{i}H+Hk_{i}.
\end{equation}

First, we consider the commutator 
\begin{equation}
\begin{split}
\lbrace H_{i},H\rbrace &=\Lambda(\rd H_{i},\rd H)+H_{i}EH-HEH_{i}\\
&=\Lambda(\bs{K}_{i}^{T}\rd H+H\gamma_{i},\rd H)+H_{i}EH-HEH_{i}\\
&=\Lambda(\bs{K}_{i}^{T}\rd H,\rd H)+H\Lambda(\gamma_{i},\rd H)+H_{i}EH-HEH_{i}\\
&=\Lambda(\bs{K}_{i}^{T}\rd H,\rd H)+H((\bs{K}_{i}^{T}\rd H)(E)-k_{i}(\rd H)(E))+H_{i}EH-HEH_{i}\\
&=\Lambda(\bs{K}_{i}^{T}\rd H,\rd H)+H(\bs{K}_{i}^{T}\rd H)(E)-Hk_{i}EH+H_{i}EH-HEH_{i}.
\end{split}
\end{equation}
Taking into account that
\begin{equation}
(\bs{K}_{i}^{T}\rd H)(E)=(\rd H_{i}-H\gamma_{i})(E)=\rd H_{i}-H\gamma_{i}(E)=EH_{i};
\end{equation}
we obtain
\begin{equation}
\begin{split}
\lbrace H_{i},H\rbrace &=\Lambda(\bs{K}_{i}^{T}\rd H,\rd H)+HEH_{i}-Hk_{i}EH+H_{i}EH-HEH_{i}\\
&=\Lambda(\bs{K}_{i}^{T}\rd H,\rd H)-Hk_{i}EH+H_{i}EH\\
&=\Lambda(\bs{K}_{i}^{T}\rd H,\rd H)-(H_{i}-Y_{i}H)EH+H_{i}EH\\
&=\Lambda(\bs{K}_{i}^{T}\rd H,\rd H)+Y_{i}HEH.
\end{split}
\end{equation}
From the first equation of system (\ref{systemcc}) we deduce
\begin{equation}
\begin{split}
\Lambda(\bs{K}_{i}^{T}\rd H,\rd H)+Y_{i}HEH&=\Lambda(\rd H,\bs{K}_{i}^{T}\rd H)-Y_{i}HEH\\
&=-\Lambda(\bs{K}_{i}^{T}\rd H,\rd H)-Y_{i}HEH,
\end{split}
\end{equation}
which implies that $$\Lambda(\bs{K}_{i}^{T}\rd H,\rd H)+Y_{i}HEH=0.$$  Consequently,
\begin{equation}
\lbrace H_{i},H\rbrace=0.
\end{equation}
Let us prove now that $\lbrace H_{i},H_{j}\rbrace=0$. To this aim, observe that
\begin{equation}
\begin{split}
\lbrace H_{i},H_{j}\rbrace &=\Lambda(\rd H_{i},\rd H_{j})+H_{i}EH_{j}-H_{j}EH_{i}\\
&=\Lambda(\bs{K}_{i}^{T}\rd H+H\gamma_{i},\bs{K}_{j}^{T}\rd H+H\gamma_{j})+H_{i}EH_{j}-H_{j}EH_{i}\\
&=\Lambda(\bs{K}_{i}^{T}\rd H,\bs{K}_{j}^{T}\rd H)+H\Lambda(\bs{K}_{i}^{T}\rd H,\gamma_{j})+H\Lambda(\gamma_{i},\bs{K}_{j}^{T}\rd H)\\
&\quad+H^{2}\Lambda(\gamma_{i},\gamma_{j})+H_{i}EH_{j}-H_{j}EH_{i}\\
&=\Lambda(\bs{K}_{i}^{T}\rd H,\bs{K}_{j}^{T}\rd H)+H(k_{j}(\bs{K}_{i}^{T}\rd H)(E)-(\bs{K}_{j}^{T}\bs{K}_{i}^{T}\rd H)(E))\\
&\quad+H((\bs{K}_{i}^{T}\bs{K}_{j}^{T}\rd H)(E)-k_{i}(\bs{K}_{j}^{T}\rd H)(E))+H^{2}((\bs{K}_{i}^{T}\gamma_{j})(E)\\
&\quad-k_{i}\gamma_{j}(E))+H_{i}EH_{j}-H_{j}EH_{i}\\
&=\Lambda(\bs{K}_{i}^{T}\rd H,\bs{K}_{j}^{T}\rd H)+H(k_{j}EH_{i}-k_{i}EH_{j}+\rd H(\bs{K}_{j}\bs{K}_{i}E-\bs{K}_{i}\bs{K}_{j}E))\\
&\quad+H^{2}\gamma_{j}(\bs{K}_{i}E)+H_{i}EH_{j}-H_{j}EH_{i}\\
&=\Lambda(\bs{K}_{i}^{T}\rd H,\bs{K}_{j}^{T}\rd H)+H_{j}EH_{i}-Y_{j}HRH_{i}-H_{i}EH_{j}+Y_{i}HEH_{j}\\
&\quad+H\rd H(\bs{K}_{j}\bs{K}_{i}E-\bs{K}_{i}\bs{K}_{j}E))+H^{2}\gamma_{j}(\bs{K}_{i}E)+H_{i}EH_{j}-H_{j}EH_{i}\\
&=\Lambda(\bs{K}_{i}^{T}\rd H,\bs{K}_{j}^{T}\rd H)-Y_{j}HRH_{i}+Y_{i}HEH_{j}\\
&\quad+H\rd H(\bs{K}_{j}\bs{K}_{i}E-\bs{K}_{i}\bs{K}_{j}E))+H^{2}\gamma_{j}(\bs{K}_{i}E).
\end{split}
\end{equation}
Since $\mathscr{H}$ is an Abelian algebra, we necessarily have $\bs{\mathcal{K}}_{i}\bs{\mathcal{K}}_{j}=\bs{\mathcal{K}}_{j}\bs{\mathcal{K}}_{i}$. For all $(X,f)\in\mathfrak{X}(M)\times C^{\infty}(M)$ we obtain
\begin{equation}
\bs{\mathcal{K}}_{i}\bs{\mathcal{K}}_{j}(X,f)=(\bs{K}_{i}\bs{K}_{j}X+Y_{i}\gamma_{j}X+f\bs{K}_{i}Y_{j}+fk_{j}Y_{i},(\bs{K}_{j}^{T}\gamma_{i})(X)+k_{i}\gamma_{j}X+f\gamma_{i}Y_{j}+fk_{i}k_{j}).
\end{equation}
Explicitly, 
\begin{equation}
\bs{\mathcal{K}}_{i}\bs{\mathcal{K}}_{j}=(\bs{K}_{i}\bs{K}_{j}+Y_{i}\otimes\gamma_{j},\bs{K}_{i}Y_{j}+k_{j}Y_{i},\bs{K}_{j}^{T}\gamma_{i}+k_{i}\gamma_{j},\gamma_{i}Y_{j}+k_{i}k_{j}),
\end{equation}
where $Y_{i}\otimes\gamma_{j}$ is the (1,1)-tensor field on $M$ defined by $(Y_{i}\otimes\gamma_{j})(X):=Y_{i}\gamma_{j}X$. By taking $X=E$ we obtain
\begin{equation} \label{eq:4.16}
\left\lbrace \begin{array}{c}
\bs{K}_{i}\bs{K}_{j}E=\bs{K}_{j}\bs{K}_{i}E,\\
\gamma_{i}(\bs{K}_{j}E)=\gamma_{j}(\bs{K}_{i}E),\\
\gamma_{i}Y_{j}=\gamma_{j}Y_{i}.
\end{array} \right. 
\end{equation}
Now, as $\lbrace H_{i},H_{j}\rbrace=-\lbrace H_{j},H_{i}\rbrace$, we have that  $\gamma_{i}(\bs{K}_{j}E)=-\gamma_{j}(\bs{K}_{i}E)$. Combining this relation with the second one of system \eqref{eq:4.16}, we deduce that $\gamma_{i}(\bs{K}_{j}E)=0$. Consequently, the expression of the commutator $\lbrace H_{i},H_{j}\rbrace$ reduces to
\begin{equation}
\lbrace H_{i},H_{j}\rbrace=\Lambda(\bs{K}_{i}^{T}\rd H,\bs{K}_{j}^{T}\rd H)-Y_{j}HRH_{i}+Y_{i}HEH_{j}.
\end{equation}

On the other hand, $\bs{\mathcal{K}}_{i}\bs{\mathcal{K}}_{j}=: \bs{\mathcal{K}}_{s}$ is an element of the algebra $\mathscr{H}$. Thus,
\begin{equation}
\Lambda(\bs{K}_{s}^{T}\rd H,\rd H)+Y_{s}HEH=0.
\end{equation}
Being $\bs{K}_{s}=\bs{K}_{i}\bs{K}_{j}+Y_{i}\otimes\gamma_{j}$ and $Y_{s}=\bs{K}_{i}Y_{j}+k_{j}Y_{i}$, we can easily prove that
\begin{equation}
\Lambda(\bs{K}_{s}^{T}\rd H,\rd H)+Y_{s}HEH=\Lambda(\bs{K}_{i}^{T}\rd H,\bs{K}_{j}^{T}\rd H)-Y_{j}HRH_{i}+Y_{i}HEH_{j},
\end{equation}
i.e.,
\begin{equation}
\lbrace H_{i},H_{j}\rbrace=\Lambda(\bs{K}_{s}^{T}\rd H,\rd H)+Y_{s}HEH=0.
\end{equation}
\end{proof}

\subsection{From Jacobi-Haantjes to Poisson-Haantjes structures}
There exists a canonical procedure which allows one to associate, with each Jacobi manifold $(M, \Lambda, E)$, a homogeneous Poisson manifold. In fact, let us introduce the manifold $\tilde{M}:= M\times \mathbb{R}$ and the bivector on $\tilde{M}$ defined by
\[
\tilde{P}:=e^{-\tau}(\Lambda+ \partial_\tau \wedge E),
\]
where $\tau$ is the global coordinate on $\mathbb{R}$. By construction, $\tilde{P}$ has a vanishing Schouten-Nijenhuis bracket and is therefore a Poisson bivector. Consequently, $(\tilde{M},\tilde{P})$ is a Poisson manifold. 

\begin{remark}
Due to Proposition \ref{prop1}, $\bs{\mathcal{K}}=(\bs{K}, Y, \gamma, k)$ is a Haantjes operator on $\mathfrak{X}(M)\times C^{\infty}(M, \mathbb{R})$ if and only if $\tilde{\bs{K}}=\bs{K}+Y\otimes d\tau+\frac{\partial}{\partial \tau}\otimes\gamma+k\frac{\partial}{\partial \tau}\otimes d\tau$ is a Haantjes operator on $M\times\mathbb{R}$. In addition, coherently with the case of Jacobi-Nijenhuis structures \cite{PNJGP2003}, one can show that $\mathcal{\bs{K}}$ is compatible with the Jacobi structure $(\Lambda,E)$ on $M$ if and only if $\tilde{\bs{K}}$ satisfies $\tilde{\bs{K}}\circ\tilde{P}^{\sharp}=\tilde{P}^{\sharp}\circ\tilde{\bs{K}}^{T}$, i.e., $\tilde{\bs{K}}$ is compatible with the Poisson structure $\tilde{P}$ on $\tilde{M}$. Thus, an extended Haantjes algebra on $(M,\Lambda,E)$ of operators compatible with the Jacobi structure, leads to a Haantjes algebra on $(\tilde{M},\tilde{P})$ of operators compatible with the Poisson structure.  Equivalently, we can state that each Jacobi-Haantjes manifold has associated a Poisson-Haantjes manifold (see \cite{T2018TMP} for a study of Poisson-Haantjes manifolds). 
\end{remark}

We can go even further: starting from a dissipative Hamiltonian system on a Jacobi manifold,  we are able to lift the dynamics to a conservative Hamiltonian system on a Poisson manifold, such that the latter dynamics projects onto the original dissipative dynamics. Specifically, let $H\in C^{\infty}(M)$ and consider the Hamiltonian system $(M,\Lambda,E,H)$, its dynamics is determined by the associated Hamiltonian vector field $X_{H}=\Lambda^{\sharp}(dH)-HE.$ By defining
\begin{equation}
\tilde{H}:=e^{\tau}H,
\end{equation}
the triple $(\tilde{M},\tilde{P},\tilde{H})$ is a Hamiltonian system with dynamics determined by the Hamiltonian vector field
\begin{equation}
X_{\tilde{H}}=X_{H}+EH\frac{\partial}{\partial \tau}.
\end{equation}
In fact, we have
\begin{equation}
\begin{split}
\tilde{P}^{\sharp}(d\tilde{H})&=e^{-\tau}(\Lambda+ \partial_\tau \wedge E)^{\sharp}(e^{\tau}Hd\tau+e^{\tau}dH)\\
&=\Lambda^{\sharp}dH-HE+EH\frac{\partial}{\partial \tau}\\
&=X_{H}+EH\frac{\partial}{\partial \tau}.
\end{split}
\end{equation}
We recover the dissipative Hamiltonian dynamics on $M$ by projecting the dynamics on $\tilde{M}$ onto the immersed submanifold determined by $\tau=0$, indeed, $M\times\lbrace 0\rbrace\equiv M$, $\tilde{H}\vert_{\tau=0}\equiv H$ and $X_{\tilde{H}}\vert_{\tau=0}\equiv X_{H}$.
\vspace{2mm}

A crucial fact is that extended Haantjes chains on $M$ lift naturally to Haantjes chains on $\tilde{M}$. Let us suppose that $H$ generates an extended Haantjes chain of length $m$ on $M$ with potential functions $H_{1},\ldots,H_{m}$. Then, $\tilde{H}$ generates a Haantjes chain of the same length $m$ on $\tilde{M}$, with potential functions $\tilde{H}_{1}=e^{\tau}H_{1},\ldots ,\tilde{H}_{m}=e^{\tau}H_{m}$. Indeed, let $\{\bs{\mathcal{K}}_1,\ldots, \bs{\mathcal{K}}_m\}$ be a distinguished basis of the extended Haantjes chain generated by $H$, such that
\begin{equation}
 (\rd H_{i}, H_{i})= \bs{\mathcal{K}}_i^{T}(\rd H,H),
\end{equation}
or equivalently
\begin{equation}
\rd H_{i}=\bs{K}_{i}^{T}\rd H+H\gamma_{i}\quad\textit{and}\quad H_{i}=Y_{i}H+Hk_{i}.
\end{equation}
Then, we have 
\begin{equation}
\begin{split}
\tilde{\bs{K}}_{i}^{T}d\tilde{H}&= \tilde{\bs{K}}^{T}(e^{\tau}Hd\tau+e^{\tau}dH)\\
&=e^{\tau}(\bs{K}_{i}^{T}dH+Y_{i}Hd\tau+H\gamma_{i}+k_{i}Hd\tau)\\
&=e^{\tau}(dH_{i}+H_{i}d\tau)\\
&=d\tilde{H}_{i}.
\end{split}
\end{equation}
By using our previous construction we can combine Theorem \ref{themain} and Lemma 1 in \cite{T2018TMP}. We can summarize the discussion in the following
\begin{remark}
\label{remarkJHPH}
Let $(M,\Lambda,E,H)$ be a dissipative Hamiltonian system on $M$.
If $H\in C^{\infty}(M)$ generates an extended Haantjes chain of length $m$, endowing $M$ with a Jacobi-Haantjes structure, then the dissipated quantities $H_{1},\ldots,H_{m}$, which are the potential functions of the chain, lift to conserved quantities of the corresponding lifted conserved Hamiltonian system $(\tilde{M},\tilde{P},\tilde{H})$. These lifted quantities become the potential functions of the Haantjes chain generated by $\tilde{H}$, which in turn endows $\tilde{M}$ with the canonical Poisson-Haantjes structure naturally induced from the Jacobi-Haantjes structure on $M$.
\end{remark}

\section{Contact Hamiltonian systems and Jacobi-Haantjes structures}
\label{sec5}

In this section we study Jacobi-Haantjes structures under the framework of contact manifolds. The relation between Jacobi-Haantjes structures and the integrability of contact Hamiltonian systems is a main topic of study.

\subsection{Contact-Haantjes manifolds}
Let $(M,\theta)$ be a contact manifold of dimension $2n+1$, and let $(\Lambda,E)$ be the Jacobi structure on $M$ defined by $\theta$, namely,
\begin{equation}
\Lambda(\alpha,\beta)=\rd \theta(\sharp(\alpha),\sharp(\beta))\quad\text{and}\quad E=R=\sharp(\theta),
\end{equation}
where $\sharp$ is the inverse of the map $\flat:\mathfrak{X}(M)\longrightarrow\Omega^{1}(M)$ defined by $\flat(X)=\iota_{X} \rd \theta +(\iota_{X} \theta) \theta.$ The vector bundle map $(\Lambda, E)^{\sharp}: T^{*}M\times \mathbb{R}\to TM\times \mathbb{R}$ has the form
\begin{equation}
(\Lambda, E)^{\sharp}(\alpha,f)=(\sharp(\alpha)-\alpha(R)R-fR,\alpha(R)).
\end{equation}

\vspace{2mm}

In order to specialize our study to Hamiltonian dynamics on contact manifolds, we propose the following
\begin{definition}
\label{decH}
A contact-Haantjes (cH) manifold of class $m$ is a triple $(M,\theta,\mathscr{H})$, where $(M,\theta)$ is a contact manifold and $\mathscr{H}$ is an Abelian extended Haantjes algebra of rank $m$ of operators on $M$ of the form $\bs{\mathcal{K}}=(\bs{K}, Y, \gamma, k)$,  such that the following compatibility relations
\begin{equation}
\label{eqcH}
\left\lbrace \begin{array}{c}
\bs{K}(\sharp(\alpha))-\alpha(R)\bs{K}R+\alpha(R)Y=\sharp(\bs{K}^{T}\alpha)-\alpha(\bs{K}R)R-\alpha(Y)R,\\
\sharp(\gamma)=kR-\bs{K}R,\\
\gamma(R)=0,
\end{array} \right. 
\end{equation}
hold for all $\alpha\in\Omega^{1}(M)$.
\end{definition}
In fact, eqs. \eqref{eqcH} express the compatibility condition \eqref{eqcompatibility} between the Jacobi structure $(M,\Lambda,E)$ and the Abelian extended Haantjes algebra $\mathscr{H}$.
\vspace{2mm}

Clearly, every contact-Haantjes manifold is a Jacobi-Haantjes manifold with the Jacobi structure naturally induced by the contact form. Definition \ref{decH} is introduced with the aim of facilitating the practical use of Jacobi-Haantjes structures in the study of dissipative systems modeled as Hamiltonian systems on contact manifolds. 

\vspace{2mm}
System \eqref{eqcH} may appear rather complicated; however, it can be significantly simplified if we consider a suitable class of operators. Specifically, if we take $\bs{\mathcal{K}}=(\bs{K},Y,0,k)$ with $\bs{K}$ diagonal (hence, a Haantjes operator) and $\gamma=0$, then system (\ref{eqcH}) reduces to 
\beq
\bs{K}R=kR
\eeq
and
\begin{equation}
\label{eqreduced}
\bs{K}(\sharp(\alpha))+\alpha(R)Y=\sharp(\bs{K}^{T}\alpha)-\alpha(Y)R \quad \forall\alpha\in\Omega^{1}(M).
\end{equation}
In addition, for an operator $\bs{\mathcal{K}}$ of this class, and taking into account that $\mathcal{H}_{\bs{K}}=0$, system (\ref{eqHaantjesex}), which ensures its Haantjes property, reduces to
\begin{equation}
\label{eqreducedHaantjesex}
\begin{cases} \hspace{2mm}
\mathcal{L}_{Y}^{\tau_{\bs{K}}}\bs{K}=-Y\otimes dk, \\ \hspace{2mm}
\bs{K}^{T}dk=kdk.
\end{cases}
\end{equation}

\begin{remark}
It is worth mentioning that the kind of operators considered in the previous analysis is really relevant for our study, in fact, it includes the operators forming the natural basis of extended Haantjes algebras defining contact-Haantjes structures for completely integrable contact Hamiltonian systems (see theorem \ref{theo:LH}).
\end{remark}

\subsection{Contact-Haantjes structures and the integrability of contact Hamiltonian systems}
The notion of complete integrability for contact Hamiltonian systems has been proposed in \cite{KT2010contact, P1990TAMS,Jovanovic2015contact}. In \cite{M2003,M2005,M2014},  the existence of action--angle variables for the dynamics generated by the Reeb vector field on contact manifolds has been proved. In \cite{CLEL2025homogeneous,CLLL2023pr}, a wide extension of the classical Liouville--Arnold theorem \cite{Arnold1978} in contact geometry, and the existence of  action-angle type coordinates has been established (for related results, see also \cite{Jovanovic2012non,Zung2018conceptual}). It is worth mentioning that in \cite{B2011SIGMA}, a notion of complete integrability for contact Hamiltonian systems was proposed under the restrictive hypothesis that the Hamiltonian function be preserved by the Reeb vector field.   Following \cite{CLEL2025homogeneous,CLLL2023pr}, we adopt a more general notion of integrability, also valid for dissipative systems.
\begin{definition}
\label{def:20}
A contact Hamiltonian system $(M,\theta,H)$ on a (co-oriented) contact manifold $M$, with $dim(M)=2n+1$, is said to be completely integrable if there exist $n+1$ functions $H_{1},\ldots,H_{n+1}$ (including $H$) in involution with respect to the bracket $\{\cdot,\cdot\}_{\theta}$, called integrals, such that $\rank\, \langle \rd H_{1},\ldots,\rd H_{n+1}\rangle \geq n$. 
\end{definition} 

As one of the main results of our theory, we propose a theorem which establishes the complete integrability of a contact Hamiltonian system in the geometric framework of contact-Haantjes structures. This result is analogous to the Liouville--Haantjes theorem of the theory of symplectic-Haantjes manifolds (theorem 39 in \cite{TT2022AMPA}). In order to formulate it, we premise the concept of nondegenerate contact Hamiltonian system.
\begin{definition}
A completely integrable contact Hamiltonian system $(M,\theta,H)$ on a contact manifold $M$, with $\dim M=2n+1$ is said to be nondegenerate with respect to a set of action--angle variables $\{(J_k, \phi_k, Z)\}$ on $M$ if 
\beq
det \bigg( \frac{\partial \nu_k}{\partial J_i} \bigg)= \det \bigg( \frac{\partial^2 H}{\partial J_i \partial J_k} \bigg)\neq 0,
\eeq
where $\nu_k(J):= \frac{\partial H}{\partial J_i}$ are the associated frequencies of the system.
\end{definition}

We can now state the central result of this section.
\begin{theorem}[Liouville--Haantjes theorem for completely integrable contact Hamiltonian systems]
\label{theo:LH}
Let $(M,\theta)$ be a contact manifold with $dim(M)=2n+1$.
Assume that $(M,\theta,\mathscr{H})$ is a contact-Haantjes manifold of class $n+1$. If $H\in C^{\infty}(M)$ generates an extended Haantjes chain of length $n+1$, then the contact Hamiltonian system $(M,\theta,H)$ is completely integrable.

Conversely, for every completely integrable and non-degenerate contact Hamiltonian system $(M,\theta,H)$ there exists an Abelian extended Haantjes algebra $\mathscr{H}$ of rank $n+1$ such that $(M,\theta,\mathscr{H})$ is a contact-Haantjes manifold and $H$ generates an extended Haantjes chain of length $n+1$. 
\end{theorem}
\begin{proof}
Since every contact-Haantjes manifold is a Jacobi-Haantjes manifold, then the first statement is a corollary of Theorem \ref{themain}. Let us prove the converse statement. Given a completely integrable contact Hamiltonian system $(M,\theta,H)$ admitting $\big(H_{0}=H,H_{1},\ldots,H_{n}\big)$ integrals pairwise in involution such that $\rank \,\langle \rd H_{0},\rd H_{1},\ldots,\rd H_{n}\rangle=n+1,$ according to \cite{CLLL2023pr} there exist coordinates $(\phi^{1},\ldots,\phi^{n},J_{1},\ldots,J_{n},Z)$, where $(J_{1},\ldots,J_{n})$ are said to be action coordinates, while $(\phi^{1},\ldots,\phi^{n})$ are angle coordinates, such that the contact Hamilton equations of motion for the Hamiltonian functions $H_{j}$ are
\begin{equation}
\dot{\phi^{i}}=N^{i}_{j}, \hspace{1cm}
\dot{J_{i}}=0 , \hspace{1cm}
\dot{z}=N^{0}_{j}, \qquad  j=1, \ldots, n+1.
\end{equation}
Here $N^{0}_{j},N^{1}_{j},\ldots,N^{n}_{j}$ are functions on $M$ depending only on the coordinates $(J_{1},\ldots,J_{n})$. This is equivalent to saying that $H_{0}=H,H_{1},\ldots,H_{n}$ depend only on $(J_{1},\ldots,J_{n})$.

As a basis for the extended Haantjes algebra we wish to construct, we can consider the basis given by the \textit{extended diagonal operators}
$\lbrace\bs{\mathcal{K}}_{0},\bs{\mathcal{K}}_{1},\ldots,\bs{\mathcal{K}}_{n}\rbrace$, where $\bs{\mathcal{K}}_{0}$ is the extended identity operator, locally expressed by
\begin{equation}
\bs{\mathcal{K}}_{0}=\left(\sum_{i=1}^{n}\left(\frac{\partial}{\partial \phi^{i}}\otimes \rd \phi^{i}+\frac{\partial}{\partial J_{i}}\otimes \rd J_{i}\right)+\frac{\partial}{\partial Z}\otimes \rd Z,0,0,1\right),
\end{equation}
whereas the remaining operators read
\begin{equation}
\bs{\mathcal{K}}_{j}=\left(\sum_{i=1}^{n}\frac{ \nu_i^{j}(J)}{\nu_i (J)}\left(\frac{\partial}{\partial \phi^{i}}\otimes \rd \phi^{i}+\frac{\partial}{\partial J_{i}}\otimes \rd J_{i}\right),\sum_{i=1}^{n}\left(\frac{H_{j}}{\nu_i(J)}\right)\frac{\partial}{\partial J_{i}},0,0\right),
\end{equation}
for $j=1,\ldots,n$, which are well-defined due to the nondegeneracy assumption. These operators are Haantjes operators, indeed, it is easy to verify that system (\ref{eqreducedHaantjesex}) is fulfilled (note that it suffices to check that $\mathcal{L}_{Y}^{\tau_{\bs{K}}}\bs{K}=0$). Besides, $\rank \,\langle \rd H_{0},\rd H_{1},\ldots,\rd H_{n}\rangle=n+1$, so they define an Abelian extended Haantjes algebra $\mathscr{H}$ of rank $n+1$. It is easy to see that the operators in $\mathscr{H}$ satisfy $\gamma_{j}=0$, $\bs{K}_{j}R=k_{j}R$ and equation (\ref{eqreduced}), therefore they satisfy system (\ref{eqcH}); in turn, $(M,\theta,\mathscr{H})$ is a contact-Haantjes manifold and $H$ generates and extended Haantjes chain $\mathscr{C}$ with potential functions $H_{0}=H,H_{1},\ldots,H_{n}$. 
\end{proof}
\begin{example}
Consider the contact Hamiltonian system with $n=1$ and Hamiltonian function $H(q,p,z)=p-z$. It is shown in \cite{CLLL2023pr,CLEL2025homogeneous}  that this system is completely integrable. Let us consider the Haantjes operators
\begin{equation}
\bs{\mathcal{K}}_{1}=\bigg(\frac{\partial}{\partial q}\otimes dq+\frac{\partial}{\partial p}\otimes dp+\frac{\partial}{\partial z}\otimes dz,0,0,1 \bigg)
\end{equation}
and
\begin{equation}
\bs{\mathcal{K}}_{2}=\bigg(\frac{\partial}{\partial q}\otimes dq+\frac{\partial}{\partial p}\otimes dp,p\frac{\partial}{\partial p},0,0\bigg).
\end{equation}
We have that $\mathscr{H}=\langle \bs{\mathcal{K}}_{1},\bs{\mathcal{K}}_{2}\rangle$ is an Abelian extended Haantjes algebra of rank $2$, $(M,\theta,\mathscr{H})$ is a contact-Haantjes manifold, and $H$ generates an extended Haantjes chain $\mathscr{C}$ with the distinguished basis $\lbrace \bs{\mathcal{K}}_{1},\bs{\mathcal{K}}_{2} \rbrace$ and potential functions $H_{1}=H,H_{2}(q,p,z)=p$, which are independent integrals pairwise in involution (as shown in \cite{CLLL2023pr,CLEL2025homogeneous}).
\end{example}

\section{A theory of separation of variables for dissipative systems}
\label{sec6}

In this section, we develop a novel approach to the problem of constructing separation variables for Hamiltonian dissipative systems. 

\subsection{Darboux coordinates in contact geometry and its symplectization}

As is well-known, given a contact manifold $(M,\theta)$, there exists associated a symplectic manifold \cite{ibanez1997co}: in fact, the Poisson tensor $\tilde{P}$ on $\tilde{M}=M\times\mathbb{R}$ induces naturally the symplectic form
\begin{equation} \label{eq:omegaM}
\omega=e^{\tau}(d\theta+d\tau\wedge \theta).
\end{equation}
We start our analysis with the following
\begin{lemma}
The local coordinates $(q^{i},p_{i},z)$ are canonical for the contact form $\theta$ if and only if the local coordinates
\begin{equation}
\label{coordinateT}
\left\lbrace\begin{array}{c}
Q^{i}=e^{\tau}q^{i},\\
P_{i}=p_{i},\\
T=\tau,\\
P_{T}=e^{\tau}(z-\displaystyle{\sum_{i=1}^{n}}q^{i}p_{i}),
\end{array}\right.
\end{equation}
are canonical for the symplectic form \eqref{eq:omegaM}.
\end{lemma}
\begin{proof}
The coordinates $(q^{i},p_{i},z)$ are canonical for the contact form $\theta$ if and only if its local expression is $\theta=dz-p_{i}dq^{i}$. In the same coordinates, $\omega$ takes the form
\begin{equation}
\omega= e^{\tau}\Big(\sum_{i=1}^{n}dq^{i}\wedge dp_{i}+d\tau\wedge dz-\sum_{i=1}^{n}p_{i}d\tau\wedge dq^{i}\Big).
\end{equation}
In the new coordinates defined by relations \eqref{coordinateT}, we obtain
\begin{equation}
\begin{split}
\sum_{i=1}^{n}dQ^{i}\wedge dP_{i}+dT\wedge dP_{T}&=\sum_{i=1}^{n}\left(e^{\tau}dq^{i}\wedge dp_{i}+e^{\tau}d\tau\wedge dp_{i}+d\tau\wedge(e^{\tau}dz-e^{\tau}q^{i}dp_{i}-p_{i}dq^{i})\right)\\
&=e^{\tau}\Big(\sum_{i=1}^{n}dq^{i}\wedge dp_{i}+d\tau\wedge dz-\sum_{i=1}^{n}p_{i}d\tau\wedge dq^{i}\Big).
\end{split}
\end{equation}
\end{proof}
\noi The inverse transformation of transformation (\ref{coordinateT}) reads
\begin{equation}
\label{coordinateinvT}
\left\lbrace\begin{array}{c}
q^{i}=e^{-T}Q^{i},\\
p_{i}=P_{i},\\
\tau=T,\\
z=e^{-T}\Big(P_{T}+\displaystyle{\sum_{i=1}^{n}}Q^{i}P_{i}\Big). 
\end{array}\right.
\end{equation}
As proved in \cite{TT2022AMPA}, given a symplectic-Haantjes manifold $(M,\omega,\mathscr{H})$, there exist canonical coordinates (Darboux coordinates) for $\omega$ in which all the Haantjes operators in $\mathscr{H}$ take simultaneously a block-diagonal form. Due to their twofold role, they are called \textit{Darboux–Haantjes coordinates}. If additionally, the Haantjes algebra $\mathscr{H}$ is semisimple --each Haantjes tensor is diagonalizable--, then on a set of Darboux-Haantjes coordinates, all Haantjes tensors of $\mathscr{H}$ take simultaneously a diagonal form.
\vspace{2mm}

The \textit{Jacobi-Haantjes} theorem proved in \cite{RTT2022CNS} states that for a given symplectic-Haantjes manifold $(M,\omega,\mathscr{H})$ of class $m$, where $\mathscr{H}$ a semisimple algebra, with $dim(M)=2m$,  if $H\in C^{\infty}(M)$ generates a Haantjes chain of length $m$ with potential functions $H_{1},\ldots,H_{n}$, then each set of Darboux-Haantjes coordinates for $(M,\omega,\mathscr{H})$ are separation coordinates of the Hamilton--Jacobi equation associated with each Hamiltonian $H_{j}$ for $j=1,\ldots n$. In particular, they are separation coordinates for the Hamilton's equations of motion.

\subsection{Main theorem on SoV in contact geometry}
Our purpose is to extend the previous approach to \textit{contact geometry}. Precisely, we wish to construct a theoretical framework allowing us to solve the problem of determining separation coordinates for the contact Hamilton's equations of motion of completely integrable contact Hamiltonian systems. 
\vspace{2mm}

Let $(M,\theta,\mathscr{H})$ be a contact-Haantjes manifold of class $n+1$ with $dim(M)=2n+1$. We assume that $H\in C^{\infty}(M)$ generates an extended Haantjes chain of length $n+1$ with potential functions $H_{1}=H,\ldots,H_{n+1}$, so that for every $j\in\lbrace 1,\ldots,n+1\rbrace$ the contact Hamiltonian system $(M,\theta,H_{j})$ is completely integrable. Let $(\tilde{M},\omega,\tilde{H}_{j})$ be the lift of $(M,\theta,H_{j})$, where $\tilde{M}=M\times\mathbb{R}$, $\omega=e^{\tau}(d\theta+d\tau\wedge \theta)$ and $\tilde{H}_{j}=e^{\tau}H_{j}$. As stated in Remark \ref{remarkJHPH}, we have that $(\tilde{M},\omega,\tilde{\mathscr{H}})$ is a symplectic-Haantjes manifold of class $n+1$,  where $\tilde{\mathscr{H}}$ is the lift of the algebra $\mathscr{H}$ given by
\begin{equation}
\label{Ktildecontact}
\tilde{\bs{K}}=\bs{K}+Y\otimes d\tau+\frac{\partial}{\partial \tau}\otimes\gamma+k\frac{\partial}{\partial \tau}\otimes d\tau
\end{equation}
for each $\bs{\mathcal{K}}=(\bs{K}, Y, \gamma, k)\in\mathscr{H}$. Also, $\tilde{H}=e^{\tau}H$ generates a Haantjes chain of length $n+1$, so that the (conservative) Hamiltonian system $(\tilde{M},\omega,\tilde{H}_{j})$ is completely integrable.
Our main result is the following
\begin{theorem} \label{theo:4}
Let $(Q^{i},T,P_{i},P_{T})$, $i=1,\ldots, n$ be a system of Darboux-Haantjes coordinates for $(\tilde{M},\omega,\tilde{\mathscr{H}})$. Then the corresponding canonical coordinates $(q^{i},p_{i},z)$ on $(M,\theta)$, given by (\ref{coordinateinvT}), are separation coordinates for the contact Hamilton's equations associated with each potential function $H_{j}$, $j=1,\ldots,n+1$. 
\end{theorem}
\begin{proof}
According to the Jacobi-Haantjes theorem \cite{RTT2022CNS}, if $(Q^{i},T,P_{i},P_{T})$ are Darboux-Haantjes coordinates for $(\tilde{M},\omega,\tilde{\mathscr{H}})$, then they are separation coordinates for the (conservative) Hamilton's equations for each $\tilde{H}_{j}$. In fact, the (conservative) Hamilton's equations have the form
\begin{equation}
\left\lbrace\begin{array}{c}
\dot{Q}^{i}=F^{j}_{i}(P_{i}),\\
\dot{T}=F^{j}_{T}(P_{T}),\\
\dot{P}_{i}=0,\\
\dot{P}_{T}=0,
\end{array}\right. 
\end{equation}
for some functions $F^{j}_{1},\ldots,F^{j}_{n},F^{j}_{\tau}\in C^{\infty}M$.
Now, taking into account the coordinate transformation given by (\ref{coordinateinvT}), we obtain
\begin{equation}
\left\lbrace\begin{array}{c}
\dot{q}^{i}=-\dot{\tau}q^{i}+e^{-\tau}F^{j}_{i}(p_{i}),\\
\dot{p}_{i}=0,\\
\dot{z}=-\dot{\tau}z+e^{-\tau}\displaystyle{\sum_{i=1}^{n}}p_{i}F^{j}_{i}(p_{i}).
\end{array}\right. 
\end{equation}
By projecting the dynamics onto the manifold $\tau=0$, therefore $\dot{\tau}=c\in\mathbb{R}$, we have that the dynamical equations on $M$ for the Hamiltonian function $H_{j}\equiv\tilde{H}_{j}\vert_{\tau=0}$ are
\begin{equation}
\left\lbrace\begin{array}{c}
\dot{q}^{i}=F^{j}_{i}(p_{i})-cq^{i},\\
\dot{p}_{i}=0,\\
\dot{z}=-cz+\displaystyle{\sum_{i=1}^{n}}p_{i}F^{j}_{i}(p_{i}).
\end{array}\right. 
\end{equation}
\end{proof}

\section{New completely integrable dissipative Hamiltonian systems and their associated Jacobi-Haantjes structures}
\label{sec7}

We present some examples of nontrivial completely integrable dissipative Hamiltonian systems on contact manifolds, all of them new to the best of our knowledge, which are obtained in the framework of Jacobi-Haantjes geometry.

\subsection{JH manifolds for 3D systems}
Let us consider first a contact manifold $(M,\theta)$ with $\dim M=3$ and canonical coordinates $(p,q,z)$.

\subsubsection{} Let us consider the extended Haantjes algebra $\mathscr{H}_{1}$ on $M$ generated by the Haantjes operators 
\beq \label{eq:K1}
\mathcal{K}_1= \left(\frac{\partial}{\partial p}\otimes \rd p+\frac{\partial}{\partial q}\otimes \rd q+\frac{\partial}{\partial z}\otimes \rd z,0,0,1\right)
\eeq
and
\bea \nn
\mathcal{K}_2&=& \left(\frac{\partial}{\partial p}\otimes \rd p+\frac{1}{2}(q-3p^2)\frac{\partial}{\partial p}\otimes \rd z+\frac{\partial}{\partial q}\otimes \rd q+ \frac{p}{2}\frac{\partial}{\partial q}\otimes \rd z+\frac{\partial}{\partial z}\otimes \rd z,Y_{2},0,0\right).
\eea
with $Y_{2}=\frac{1}{2}p\frac{\partial}{\partial p}+(\frac{1}{2}q-\frac{3}{2}p^{2})\frac{\partial}{\partial q}+(\frac{1}{2}qp-\frac{3}{2}p^{3})\frac{\partial}{\partial z}$. We have that $(M,\theta,\mathscr{H}_{1})$ is a contact-Haantjes manifold, and that \beq
H_1=p^2-\frac{p^3}{2}+\frac{p q}{2}-z
\eeq
generates an extended Haantjes chain of length $2$ with potential functions $H_{1}$ and
\beq
H_{2}=p^2.
\eeq
Thus, $(M,\omega, \{H_1,H_2)\}$ is a completely integrable contact Hamiltonian system modeling a dissipative system.
\vspace{2mm}

Now, let   $(\tilde{M},\omega)$ be the symplectization of the contact manifold $(M,\theta)$, i.e., $\tilde{M}=M\times\mathbb{R}$ and $\omega=e^{\tau}(d\theta+d\tau\wedge \theta).$ Let us consider  canonical coordinates $(Q,T,P,P_{t})$ on $(\tilde{M},\omega)$ given by (\ref{coordinateT}). Finally, let $\tilde{\mathscr{H}}$ be the lift of the algebra $\mathscr{H}$ given by the operators (\ref{Ktildecontact}). We have that the coordinates $(\tilde{Q},\tilde{T},\tilde{P},\tilde{P}_{t})$ given by
\begin{equation}
\left\lbrace\begin{array}{c}
\tilde{Q}=\frac{3}{2}e^{T}P^{3}+\frac{1}{2}QP,\\
\tilde{T}=T,\\
\tilde{P}=2lnP+T,\\
\tilde{P}_{T}=\frac{1}{2}e^{T}P^{3}+\frac{1}{2}QP+P_{T},
\end{array}\right. 
\end{equation}
are Darboux-Haantjes coordinates for $(\tilde{M},\omega,\tilde{\mathscr{H}})$; indeed
\begin{equation}
\tilde{K}_{1}=\frac{\partial}{\partial \tilde{Q}}\otimes d\tilde{Q}+\frac{\partial}{\partial \tilde{T}}\otimes d\tilde{T}+\frac{\partial}{\partial \tilde{P}}\otimes d\tilde{P}+\frac{\partial}{\partial \tilde{P}_{T}}\otimes d\tilde{P}_{T}
\end{equation}
and
\begin{equation}
\tilde{K}_{2}=\frac{\partial}{\partial \tilde{Q}}\otimes d\tilde{Q}+\frac{\partial}{\partial \tilde{P}}\otimes d\tilde{P}.
\end{equation}
According to Theorem \ref{theo:4}, we have that the corresponding canonical coordinates $(\tilde{q},\tilde{p},\tilde{z})$ on $(M,\theta)$, given by (\ref{coordinateinvT}), are separation coordinates for the contact Hamilton's equations associated with each potential function $H_{1},H_{2}$. In fact, the expression for the lifts $\tilde{H}_{1},\tilde{H}_{2}$ of $H_{1},H_{2}$ on the coordinates $(\tilde{Q},\tilde{T},\tilde{P},\tilde{P}_{t})$ are
\begin{equation}
\tilde{H}_{1}=e^{\tilde{P}}-\tilde{P}_{T}
\end{equation}
and
\begin{equation}
\tilde{H}_{2}=e^{\tilde{P}}.
\end{equation}
Thus, on $(\tilde{q},\tilde{p},\tilde{z})$ the corresponding Hamilton's equations of motion are
\begin{equation}
\left\lbrace\begin{array}{c}
\dot{\tilde{q}}=e^{\tilde{p}}+\tilde{q},\\
\dot{\tilde{p}}=0,\\
\dot{\tilde{z}}=z+\tilde{p}e^{\tilde{p}},
\end{array}\right. 
\end{equation}
and
\begin{equation}
\left\lbrace\begin{array}{c}
\dot{\tilde{q}}=e^{\tilde{p}},\\
\dot{\tilde{p}}=0,\\
\dot{\tilde{z}}=\tilde{p}e^{\tilde{p}},
\end{array}\right. 
\end{equation}
for the Hamiltonian functions $H_{1}$ and $H_{2}$ respectively.

\subsubsection{} Another interesting case arise when considering the extended Haantjes algebra $\mathscr{H}_{2}$ generated by the operator \eqref{eq:K1} and
\bea \nn
\mathcal{K}_2 &=& \left(\frac{\partial}{\partial p}\otimes \rd p-\Big[\frac{\sin p}{2\cos^2 p }(3+\cos 2p+2 q \sin p)\Big]\frac{\partial}{\partial p}\otimes \rd z \right. \\ \nn
&+& \left. \frac{\partial}{\partial q}\otimes \rd q+ (p-\tan p)\frac{\partial}{\partial q}\otimes \rd z+\frac{\partial}{\partial z}\otimes \rd z,Y_{2},0,0\right)
\eea
with $Y_{2}=\tan p\frac{\partial}{\partial p}+(-\sin p-\sec p\tan p-q\tan^{2}p)\frac{\partial}{\partial q}+(-p\sin p+q\tan p-p\sec p\tan p-q\sec^{2}p\tan p-qp\tan^{2}p+q\tan^{3}p)\frac{\partial}{\partial z}$. The triple $(M,\theta,\mathscr{H}_{2})$ is a contact-Haantjes manifold. In this case, the trigonometric Hamiltonian function 
\beq
H_{1}=p q-z +\sin p -(q+\sin p) \tan p
\eeq
generates an extended Haantjes chain of length $2$ with potential functions $H_{1}$ and
\beq
H_{2}=\sin p.
\eeq
Consequently, $(M,\omega,\{H_1,H_2\})$ is a completely integrable contact Hamiltonian system modeling a dissipative system.

\subsection{5D systems and JH manifolds}
Let $M$ be a differentiable manifold, with $\dim M=5$. 

We introduce the extended Haantjes algebra $\mathscr{H}_3$ generated by
\beq \label{eq:iden}
\mathcal{K}_1= \left(\sum_{i=1}^{2}\frac{\partial}{\partial p_i}\otimes \rd p_i+\sum_{i=1}^{2}\frac{\partial}{\partial q_i}\otimes \rd q_i+\frac{\partial}{\partial z}\otimes \rd z,0,0,1\right)
\eeq
\bea
\mathcal{K}_2&=&  \Bigg(\frac{\partial}{\partial p_1}\otimes \rd p_1+ \frac{\partial}{\partial p_2}\otimes \rd p_2+\frac{\partial}{\partial q_1}\otimes \rd q_1+\frac{\partial}{\partial q_2}\otimes \rd q_2 \\ \nn 
&+&   (2 p_1+4 q_1-1)  \frac{\partial}{\partial z}\otimes \rd p_1+ (2 p_2+ 4 q_2) \frac{\partial}{\partial z}\otimes \rd p_2 \\ \nn
&+& (4 p_1+ 6 q_1 -2) \frac{\partial}{\partial z}\otimes \rd q_1 +  (4 p_2+ 6 q_2) \frac{\partial}{\partial z} \otimes \rd q_2, Y_{2},0,0 \Bigg)
\eea
with $Y_{2}=\frac{1}{2}\frac{\partial}{\partial p_{1}}$,
and
\bea
\nn \mathcal{K}_3&=&  \Bigg(\frac{\partial}{\partial p_1}\otimes \rd p_1 +  \frac{\partial}{\partial p_2}\otimes \rd p_2 + \frac{\partial}{\partial q_1}\otimes \rd q_1 + \frac{\partial}{\partial q_2}\otimes \rd q_2 + \\ \nn  &+& \Big(2 p_1+4 q_1+\frac{(p_2+2 q_2)^2}{(p_1+2 q_1)^2}\Big) \frac{\partial}{\partial z}\otimes \rd p_1 +\Big(2 p_2+ 4 q_2 -\frac{2(p_2+2 q_2)}{p_1+2 q_1}\Big) \frac{\partial}{\partial z}\otimes \rd p_2 \\ \nn &+&  \Big(4 p_1+ 6 q_1 +\frac{2(p_2+2 q_2)}{(p_1+2 q_1)^2}\Big) \frac{\partial}{\partial z }\otimes \rd q_1 +\Big(4 p_2+ 6 q_2-\frac{4(p_2+2 q_2)}{p_1+2 q_1}\Big) \frac{\partial}{\partial z} \otimes \rd q_2,Y_{3},0,0 \Bigg) \\
\eea
with $Y_{3}=\frac{p_{2}+2q_{2}}{2p_{1}+4q_{1}}\frac{\partial}{\partial p_{2}}$. We have that $(M,\theta,\mathscr{H}_{3})$ is a contact-Haantjes manifold, and that
\beq
H_1=p_1^2 +p_2^2+ 4 p_1 q_1 +3 q_1^2 + 4 p_2 q_2 +3 q_2^2 -z
\eeq
generates an extended Haantjes chain of length $3$ with potential functions $H_{1}$,
\bea
H_2 &=& p_1+2 q_1,  \\
H_3 &=& \frac{(p_2+2 q_2)^2}{p_1+ 2 q_1}.
\eea
Therefore, in this case we have defined a 5-dimensional completely integrable contact Hamiltonian system.

\subsubsection{} Another interesting integrable dissipative system arises from the Haantjes algebra $\mathscr{H}_4$, generated by the tensor field $\mathcal{K}_1$ given in eq. \eqref{eq:iden}, and the tensor fields
\bea \nn
\mathcal{K}_2&=& \Bigg(\frac{\partial}{\partial p_1}\otimes \rd p_1+ \frac{\partial}{\partial p_2}\otimes \rd p_2+\frac{\partial}{\partial q_1}\otimes \rd q_1+\frac{\partial}{\partial q_2}\otimes \rd q_2 \\ \nn 
&+&   \Big(2p_1(e^{p_2}-1) + \frac{1+q_1}{2}\Big)  \frac{\partial}{\partial z}\otimes \rd p_1+ \left(e^{p_2}\left(p_1^2 + \frac{p_2-2}{p_2^3}\right) + q_2\right) \frac{\partial}{\partial z}\otimes \rd p_2 \\ \nn
&+& \frac{p_1}{2} \frac{\partial}{\partial z}\otimes \rd q_1 +  p_2 \frac{\partial}{\partial z} \otimes \rd q_2, Y_2,0,0 \Bigg), \\
\eea
with $Y_{2}=\frac{1}{2e^{p_{2}}}\left(p_{1}\frac{\partial}{\partial p_{1}}+(1+q_{1})\frac{\partial}{\partial q_{1}}+(1+q_{1})p_{1}\frac{\partial}{\partial z}\right)$, and
\bea
\nn \mathcal{K}_3&=&  \Bigg(\frac{\partial}{\partial p_1}\otimes \rd p_1 +  \frac{\partial}{\partial p_2}\otimes \rd p_2 + \frac{\partial}{\partial q_1}\otimes \rd q_1 + \frac{\partial}{\partial q_2}\otimes \rd q_2 + \\ \nn  &+& \Big(2p_1 e^{p_2} + \frac{1+q_1}{2}\Big) \frac{\partial}{\partial z}\otimes \rd p_1 +\Big(e^{p_2}\left(p_1^2 + \frac{p_2-2}{p_2^3}-1\right) + q_2\Big) \frac{\partial}{\partial z}\otimes \rd p_2 \\ \nn &+&  \frac{p_1}{2} \frac{\partial}{\partial z }\otimes \rd q_1 + p_2 \frac{\partial}{\partial z} \otimes \rd q_2, Y_3,0,0 \Bigg) \\
\eea
with $Y_{3}=\frac{1}{p_{2}^{2}}\frac{\partial}{\partial p_{2}}+\left(\frac{e^{p_{2}}}{p_{1}^{2}p_{2}^{2}}-\frac{2e^{p_{2}}}{p_{1}^{2}p_{2}^{3}}+\frac{q_{2}}{p_{1}^{2}}\right)\frac{\partial}{\partial q_{2}}+\left(\frac{e^{p_{2}}}{p_{1}^{2}p_{2}}-\frac{2e^{p_{2}}}{p_{1}^{2}p_{2}^{2}}+\frac{p_{2}q_{2}}{p_{1}^{2}}\right)\frac{\partial}{\partial q_{z}}$.
\vspace{2mm}

We have that $H_1= p_1^2 e^{p_2}+ \frac{e^{p_2}}{p_2^2}+ \frac{p_1}{2} +\frac{p_1 q_1}{2}+ p_2 q_2  -z$ generates an extended Haantjes chain of length $3$ with potential functions $H_{1}$, $H_2=p_1^2$ and $H_3=e^{p_2}.$

\subsection{A new $(2n+1)$-dimensional completely integrable dissipative system}
\subsubsection{} Finally, we consider the $n+1$ dimensional Haantjes algebra $\mathscr{H}_5$ generated by the operators $\mathcal{K}_0,\mathcal{K}_1, \ldots, \mathcal{K}_n$, where
\beq
\mathcal{K}_0= \left(\sum_{i=1}^{n}\frac{\partial}{\partial p_i}\otimes \rd p_i+\sum_{i=1}^{n}\frac{\partial}{\partial q_i}\otimes \rd q_i+\frac{\partial}{\partial z}\otimes \rd z,0,0,1\right)
\eeq
and
\bea \nn
\mathcal{K}_j &=&  \Bigg(\sum_{i=1}^{n}\frac{\partial}{\partial p_i}\otimes \rd p_i + \sum_{i=1}^{n} \frac{\partial}{\partial q_i}\otimes \rd q_i + \sum_{i=1}^{n} \Big(2 p_i + \frac{1+q_i}{2} - 2p_{j-1} \delta_{j-1,i}\Big)  \frac{\partial}{\partial z}\otimes \rd p_i \\  \nn &+& \sum_{i=1}^{n} \frac{p_i}{2} \frac{\partial}{\partial z }\otimes \rd q_i, Y_j,0,0 \Bigg),
\eea
with $Y_{j}=\frac{p_{j}}{2}\frac{\partial}{\partial p_{j}}+(\frac{1+q_{j}}{2})\frac{\partial}{\partial q_{j}}+\frac{(1+q_{j})p_{j}}{2}\frac{\partial}{\partial z}$, $j=1,\ldots, n$.
\vspace{2mm}

We have that
\bea
H_0= \sum_{i=1}^{n} p_i^2+ \frac{1}{2}\sum_{i=1}^{n}\big( p_i+ p_i q_i\big) -z
\eea
generates an extended Haantjes chain of length $n+1$ with potential functions $H_{0},H_{1},\ldots,H_{h}$ given by
\bea
H_j = p_j^2,\quad j=1,\ldots,n.
\eea

\section*{Data Availability Statement}  No datasets were generated or analyzed during the current study.

\section*{Conflict of Interest Statement} The authors declare that they have no conflicts of interest.

\section*{Acknowledgement}

R. A. wishes to thank the financial support provided by the Secretaría de Ciencia, Humanidades, Tecnología e Innovación (SECIHTI) of Mexico through a postdoctoral fellowship under the Estancias Posdoctorales por México 2022 program. 


The research of P. T. has also been supported by the Project PID2024-156610NB-I00 of Ministerio de Ciencias, Innovaci\'on y Universidades,  and by  the Severo Ochoa Programme for Centres of Excellence in R\&D
(CEX-2023-001347-S), Ministerio de Ciencia, Innovaci\'{o}n y Universidades y Agencia Estatal de Investigaci\'on, Spain.

P.T. is a member of the Gruppo Nazionale di Fisica Matematica (GNFM) of the Istituto Nazionale di Alta Matematica (INdAM).

\end{document}